\documentclass[a4paper,onecolumn,11pt,unpublished]{quantumarticle}
\pdfoutput=1
\usepackage[utf8]{inputenc}
\usepackage[english]{babel}
\usepackage[T1]{fontenc}
\usepackage{amsmath,amssymb,amsthm}
\usepackage{thmtools}
\usepackage{hyperref}
\usepackage[nameinlink]{cleveref} 
\usepackage{braket}
\usepackage{bm}
\usepackage{dsfont}

\usepackage{tikz}
\usepackage{lipsum}

\usepackage[compress,numbers]{natbib}

\declaretheorem[
    name=Theorem,
    refname={theorem,theorems},
    Refname={Theorem,Theorems}
]{theorem}

\declaretheorem[
    sibling=theorem,
    name=Proposition,
    refname={proposition,propositions},
    Refname={Proposition,Propositions}
]{proposition}

\declaretheorem[
    sibling=theorem,
    name=Lemma,
    refname={lemma,lemmas},
    Refname={Lemma,Lemmas}
]{lemma}

\declaretheorem[
    style=definition,
    name=Definition,
    refname={definition,definitions},
    Refname={Definition,Definitions}
]{definition}

\declaretheorem[
    style=definition,
    sibling=definition,
    name=Fact,
    refname={fact,facts},
    Refname={Fact,Facts}
]{fact}

\declaretheorem[
    style=definition,
    sibling=definition, 
    name=Notation,
    refname={notation,notations},
    Refname={Notation,Notations}
]{notation}

\newenvironment{repeatedstatement}[2][]
{%
    \par\medskip\noindent
    \textbf{\Cref{#2}}%
    \if\relax\detokenize{#1}\relax
        \textbf{.}%
    \else
        \ \textnormal{(#1).}%
    \fi
    \ \itshape
}
{%
    \par\medskip
}

\newcommand{\bbE}{\mathbb{E}}
\newcommand{\U}{\mathcal{U}}
\newcommand{\e}{\epsilon}
\def\ketbra#1#2{\mathinner{|{#1}\rangle\langle{#2}|}}
\def\doubleket#1{\mathinner{|{#1}\rangle \hspace{-.6mm} \rangle}}

\begin{document}

\title{Near-Optimal Parallel Unitary Process Tomography in Diamond Distance}

\author{Noam Scully}
\email{nscully@uwaterloo.ca}
\affiliation{Department of Physics and Astronomy, University of Waterloo, Ontario N2L 3G1, Canada}
\affiliation{Perimeter Institute for Theoretical Physics, Waterloo, Ontario N2L 2Y5, Canada}
\author{Sisi Zhou}
\email{szhou1@perimeterinstitute.ca}
\affiliation{Perimeter Institute for Theoretical Physics, Waterloo, Ontario N2L 2Y5, Canada}
\affiliation{Department of Physics and Astronomy, University of Waterloo, Ontario N2L 3G1, Canada}
\affiliation{Department of Applied Mathematics and Institute for Quantum Computing, University of Waterloo, Ontario N2L 3G1, Canada}
\maketitle

\begin{abstract}
Unitary process tomography is the act of learning how an unknown unitary gate acts on a quantum system. It is an essential procedure for realizing a functioning quantum computer. As quantum technologies develop, it is important to understand the computational resources needed to conduct unitary process tomography using a variety of types of circuits. In particular, when measuring error in diamond distance, it is known how to achieve optimal query complexity using a sequential circuit architecture, but it had not been known whether a parallel circuit architecture could achieve the same optimal query complexity. In this work, we prove that unitary process tomography with near-optimal query complexity is indeed possible using a parallel circuit by introducing an explicit, novel parallel procedure. This closes the gap between sequential and parallel circuits up to a logarithmic factor, showing that parallel and sequential strategies are nearly equally capable for this task.
\end{abstract}

\section{Introduction}

The task of \textit{unitary process tomography} is to learn how a unitary gate acts on a system when the gate's inner workings are initially unknown. The desired output of unitary process tomography procedures is a classical description of a unitary gate that is an $\e$-close estimate of the unknown unitary, as measured by a notion of distance between unitary gates. This procedure has been studied extensively and is of fundamental importance in quantum computation~\cite{baldwin2014,gutoski2014process,scott2008optimizing,Bisio_2010,yrc,HKOT,Grewal_Liang_2026,HLSZY}. For example, process tomography can play a role in the engineering of quantum devices during gate certification~\cite{gebhart2023learning,eisert2020quantum}. It also underpins certain applications of those devices, such as sensing~\cite{gebhart2023learning}. 

A growing body of research seeks to understand the amount of resources required to carry out unitary process tomography in a variety of settings, including investigating different circuit types~\cite{Bisio_2010,yrc,HKOT,Grewal_Liang_2026,HLSZY}. In a \textit{sequential} protocol, the unknown unitary acts many times in succession on a single qudit, possibly with adaptive controls between queries. Complementarily, in a non-adaptive \textit{parallel} protocol, all queries of the unknown unitary occur simultaneously on a register of many qudits, and the unitary acts at most once on each qubit. Sequential protocols generalize parallel protocols: A sequential protocol can simulate a parallel one using the same number of queries by interleaving SWAP gates between sequential applications of the unitary~\cite{giovannetti2006quantum,liu2023optimal,bavaresco2021strict}. These circuit types have tradeoffs when it comes to implementation. Parallel circuits require more space but have the upside of low circuit depths. Meanwhile, sequential circuits may have low space overhead but require significantly longer coherent control time. To execute unitary process tomography in practice, it is beneficial to understand whether it can be efficiently achieved with various circuit types, given that we cannot predict which types will be most feasible on future quantum technologies. This question also has theoretical significance: When analyzed using the framework of quantum metrology, the complexity of unitary estimation is sensitive to the type of circuit used~\cite{IF,Yuan}. Resolving whether the same is true in the tomography framework considered in this paper could open up new questions about how these two frameworks relate to each other.

Currently, it is known that all unitary process tomography protocols must query the unknown unitary at least $\Omega(d^2\sqrt{\e^{-1}})$ times, where $d$ is the dimension of the unitary and the closeness $\e$ of the estimate is measured with entanglement infidelity~\cite{YKSQM}. It is possible to match that lower bound with a parallel algorithm that has query complexity $O(d^2 \sqrt{\e^{-1}})$ with respect to entanglement infidelity~\cite{yrc}. Since sequential circuits can simulate parallel ones, the query complexity for this task is also $O(d^2 \sqrt{\e^{-1}})$ for sequential circuits. There is no difference in query complexity between parallel and sequential unitary tomography protocols when error is measured using entanglement infidelity.

However, we are interested in measuring error in diamond distance. Diamond distance is more stringent than entanglement infidelity and captures worst-case differences between our estimate and the unknown unitary~\cite{HKOT}. When using diamond distance to measure error instead of entanglement infidelity, parallel and sequential strategies have not been previously shown to perform equally efficiently. When $\e$ is measured in diamond distance, at least $\Omega(d^2\e^{-1})$ queries are needed for unitary process tomography~\cite{HKOT}. This holds even with access to the unknown unitary's inverse and to controlled versions of those gates~\cite{HKOT}. Using a sequential algorithm, one can match that lower bound, achieving query complexity $O(d^2 \e^{-1})$ in diamond distance~\cite{HKOT,Grewal_Liang_2026}. When restricted to parallel protocols, it is possible to convert the distance measure of the infidelity-optimal protocol in~\cite{yrc} from entanglement infidelity to diamond distance~\cite{HKOT}. The resulting complexity is known to fall between $O(d^2 \e^{-1})$ and $O(d^{2.5} \e^{-1})$~\cite{HKOT}, potentially falling short of optimal. A different parallel protocol based on classical shadow estimation of unitary channels yields the same result of $O(d^{2.5} \e^{-1})$ parallel queries\footnote
{
An error occurred in the proof of Theorem 3.3 of~\cite{HLSZY}. As a result, the query complexity of the unitary tomography protocol in~\cite{HLSZY} is in fact $O(d^{2.5} \e^{-1})$, not the initially claimed $O(d^{2} \e^{-1})$.
}~\cite{HLSZY}.
These parallel protocols both exhibit \textit{Heisenberg scaling}, meaning optimal dependence on $\e$ of $\e^{-1}$. Yet in the parallel regime, no prior work has closed the query complexity gap in dependence on $d$ between the lower bound of $\Omega(d^2 \e^{-1})$ and the upper bound of $O(d^{2.5} \e^{-1})$ for diamond distance.

In this work, we close this gap up to a logarithmic factor. For each $r \in \mathbb{N}_+$, we display a parallel unitary tomography protocol with query complexity $O(rd^{2+1/(2r)}\e^{-1})$. In particular, selecting $r = \log(d)$ yields a near-optimal parallel protocol with query complexity $O(d^2 \log(d) \e^{-1})$. Our protocol is similar to the ``measure and operate'' protocols in \cite{Bisio_2010} and \cite{yrc}. We formalize our result as follows:

\begin{definition}[Diamond distance unitary process tomography problem]\label{def:unitary_tomography_problem} Let $U \in \mathrm{U}(d)$ be an unknown $d$-dimensional unitary operator and let $\e,\eta \in (0,1)$ be known error parameters. Consider a protocol $\mathcal{A}$ that uses black-box access to $\mathcal{U}(U)$ and whose output is a classical description of a unitary $\widehat{U}$ sampled from some $\mathrm{U}(d)$-valued random variable $\bm{\widehat{U}}$. We say $\mathcal{A}$ solves the \textit{diamond distance unitary process tomography problem} with query complexity $n$ if $\mathcal{A}$ applies $\mathcal{U}(U)$ $n$ times and 
\begin{equation} \mathrm{Pr}\left(d_\diamond\left(\mathcal{U}(U),\mathcal{U}(\bm{\widehat{U}})\right) \leq \e \right) \geq 1-\eta. \end{equation}
\end{definition}

\begin{theorem}[Near-optimal parallel unitary process tomography]\label{thm:near-optimal_tomography} For all $r \in \mathbb{N}_+$ there exists a protocol solving the diamond distance unitary process tomography problem with $\eta=1/3$ and query complexity $O(rd^{2+1/(2r)}\e^{-1})$ such that all queries of $\mathcal{U}(U)$ occur in parallel.
\end{theorem}

A consequence of \Cref{thm:near-optimal_tomography} is that by repeating its protocol $O(\log 1/\eta)$ times independently and in parallel, one can construct an estimator solving the unitary process tomography problem with error parameters $(\e,\eta)$ for any $\eta \in (0,1)$; see, for example, Proposition~2.4 of \cite{HKOT}. Overall, this uses $O(rd^{2+1/(2r)}\e^{-1} \log (1/\eta))$ queries of $\mathcal{U}(U)$, all of which occur in parallel.

In summary, our protocol is the first parallel unitary tomography protocol to achieve near-optimal query complexity $O(d^2\log(d)\e^{-1})$. We therefore close the gap in $d$-dependence between parallel and sequential circuits for this task up to a logarithmic factor.

\section{Preliminaries}

\subsection{Unitaries and diamond distance}

\begin{notation}[Unitaries] We adopt the following notation for working with unitaries. 

$\mathrm{U}(d)$ is the $d$-dimensional unitary group. 

$\bm{\widehat{U}}$ is a $\mathrm{U}(d)$-valued random variable denoting the output of the estimation protocol we will be considering. 

$\widehat{U} \in \mathrm{U}(d)$ is an arbitrary unitary. We think of $\widehat{U}$ as a possible specific output of the unitary estimation protocol; for example, we may refer to the probability density associated with the outcome $\bm{\widehat{U}} = \widehat{U}$. 

$\U(U)$ is the unitary channel induced by $U \in \mathrm{U}(d)$. $\U(\bm{\widehat{U}})$ is a unitary-channel-valued random variable.
\end{notation}

\begin{definition}[Diamond distance] For channels $\mathcal{E}_1,\mathcal{E}_2$ acting on density matrices $\rho$,
\begin{align*}
    d_\diamond(\mathcal{E}_1,\mathcal{E}_2) &= \| \mathcal{E}_1 - \mathcal{E}_2 \|_\diamond \\
    &= \sup_{m}\max_\rho \| (\mathcal{E}_1 \otimes I_m)\rho - (\mathcal{E}_2 \otimes I_m)\rho \|_1.
\end{align*}
\end{definition}

\begin{definition}[Error $\bm{\e}$] Let $U \in \mathrm{U}(d)$ be an unknown unitary to be estimated and $\bm{\widehat{U}}$ be the random unitary-valued output of an estimation procedure. We define the $\mathbb{R}$-valued random variable
\begin{align*}
    \bm{\e} &= d_\diamond(\U(\bm{\widehat{U}}), \U(U))
\end{align*}
to quantify the error associated with the estimates produced by the procedure. We also use the non-bolded symbol $\e \in (0,1)$ to represent an amount of error we aim for; for example, we may discuss the probability that ${\bm{\e}} \leq \e$.
\end{definition}

\subsection{Generic parallel unitary tomography strategy from \cite{yrc}}

The protocol we use for near-optimal parallel unitary tomography is an adapted version of the tomography protocol in \cite{yrc}. Here, we lay out the shared core features of both protocols. In \Cref{sec:near-optimal_protocol_defn}, we specify how our protocol departs from the one in \cite{yrc}.

\begin{definition}[Parallel black-box protocol] Let $\mathcal{E}$ be a channel acting on qudits and let $\mathcal{A}$ be an algorithm with black-box access to $\mathcal{E}$. We say $\mathcal{A}$ is \textit{parallel} if the only queries to $\mathcal{E}$ within $\mathcal{A}$ are a single instance of a tensor power $\mathcal{E}^{\otimes m}$ acting on $m$ qudits simultaneously.
\end{definition}

\begin{definition}[{Generic parallel unitary tomography strategy \cite{yrc}}]\label{def:mo_protocol} Let $U \in \mathrm{U}(d)$ be an unknown unitary whose gate $\U(U)$ we have black box access to and which is to be estimated. The following parallel protocol results in a classical description of a unitary $\widehat{U}$ serving as our estimate:
\begin{enumerate}
    \item Apply $\U(U)^{\otimes n} \otimes I^{\otimes n}$ to a chosen probe state $\rho_\text{probe}$, a density matrix on $\mathcal{H}^{\otimes n} \otimes \mathcal{H}^{\otimes n}$.
    \item Measure the resulting state with a POVM of the form $\{ M_{\widehat{U}}\ \mathrm{d}\widehat{U} \}$ where $\mathrm{d}\widehat{U}$ is the Haar measure.
    \item Return the measurement outcome $\widehat{U}$.
\end{enumerate}
For a given $U$, $p(\widehat{U}|U)$ is defined to be the probability density of obtaining estimate $\widehat{U}$ from this process.
\end{definition}

Within this framework, we now specify the POVM we use and the form that our probe state will take. These constructions are due to \cite{Bisio_2010} and \cite{yrc}. For the representation theory background, see for example~\cite{fulton2013representation}.

\begin{definition}[Young diagram, $\mathsf{Y}_n^d$, chapter 4 of \cite{fulton2013representation}] A Young diagram of size $n$ with at  most $d$ rows is a tuple $\lambda = (\lambda_1, \ldots, \lambda_d)$ of nonnegative integers with non-increasing row size $\lambda_1 \geq \cdots \geq \lambda_d$ that sum to $n$. For a nonnegative integer $n$, $\mathsf{Y}_n^d$ is the set of all Young diagrams of size $n$ with at most $d$ rows.
\end{definition}

\begin{notation}[Schur-Weyl decomposition of the Hilbert space, chapter 6 of \cite{fulton2013representation}]
Let $\mathcal{H}$ be a $d$-dimensional Hilbert space that the unknown unitary $U$ will act on. $\mathcal{H}^{\otimes n}$ decomposes as a sum of irreducible representations $\mathcal{H}_\lambda$ corresponding to elements of $\mathsf{Y}_n^d$, each with dimension $d_\lambda$ and each appearing in the decomposition with multiplicity equal to the dimension of a multiplicity space $\mathcal{M}_\lambda$, $\mathcal{H}^{\otimes n} = \bigoplus_{\lambda \in \mathsf{Y}_n^d} \mathcal{H}_\lambda \otimes \mathcal{M}_\lambda.$ Let $U_\lambda$ denote the action of $U$ on $\mathcal{H}_\lambda$. Under this decomposition, $U^{\otimes n}$ acts by $\bigoplus_{\lambda \in \mathsf{Y}_n^d} U_\lambda \otimes I_{\text{dim}\mathcal{M}_\lambda}$.
\end{notation}

\begin{definition}[$q, \mathsf{Y}_\text{probe}$] $q$ is a chosen probability distribution over $\mathsf{Y}_n^d$,
\[ q = \{q_\lambda\}_{\lambda \in \mathsf{Y}_n^d},\ \sum_{\lambda \in \mathsf{Y}_n^d} q_\lambda = 1\]
and the probe state will depend on this choice. We let $\mathsf{Y}_\text{probe} \subset \mathsf{Y}_n^d$ denote the support of $q$. 
\end{definition}

\begin{notation}[Double ket notation] For an operator $A$, we write its vectorization  as $\doubleket{A} = \sum_{i,j} \bra{i} A \ket{j} \ket{i}\ket{j}$.
\end{notation}

\begin{definition}[Probe state for protocol \cite{Bisio_2010,yrc}]\label{def:probe_state} We use $2n$ copies of the Hilbert space $\mathcal{H}$, with $U^{\otimes n}$ acting on the first $n$ along with $n$ ancillas. The probe state is
\begin{align*}
    \ket{\psi_q} = \bigoplus_{\lambda \in \mathsf{Y}_\text{probe}} \sqrt{\frac{q_\lambda}{d_\lambda}} \doubleket{I_{d_\lambda}} \otimes \ket{\eta_\lambda} \in \mathcal{H}^{\otimes 2n}
\end{align*}
where $\ket{\eta_\lambda} \in \mathcal{M}_\lambda^{\otimes 2}$ is arbitrary. ($\ket{\eta_\lambda}$ is acted on by $I_{\dim\mathcal{M}_\lambda}$ and does not encode information about $U$.)
\end{definition}

Applying $U^{\otimes n}$, which acts by $U_\lambda \otimes I_{d_\lambda}$ on $\doubleket{I_{d_\lambda}}$, produces
\[ \ket{\psi_{q,U}} = \bigoplus_{\lambda \in \mathsf{Y}_\text{probe}} \sqrt{\frac{q_\lambda}{d_\lambda}} \doubleket{U_\lambda} \otimes \ket{\eta_\lambda}. \]
Then we measure with the following POVM.
\begin{definition}[POVM for protocol \cite{yrc}]\label{def:protocol_measurement} Letting $\mathrm{d}\widehat{U}$ denote the Haar measure on $\mathrm{U}$(d), the POVM is
\[ \{ \ketbra{\psi_{\widehat{U}}}{\psi_{\widehat{U}}} \mathrm{d}\widehat{U} \} \]
where $\ket{\psi_{\widehat{U}}} = \bigoplus_{\lambda \in \mathsf{Y}_\text{probe}} \sqrt{d_\lambda} \doubleket{\widehat{U}_\lambda} \otimes \ket{\eta_\lambda}$.
When using this probe state and measurement in the protocol in \Cref{def:mo_protocol}, the probability density of getting measurement outcome $\widehat{U}$ when the target unitary is $U$ is~\cite{yrc}
\[ p(\widehat{U}|U) = \left| \sum_{\lambda \in \mathsf{Y}_\text{probe}} \sqrt{q_\lambda} \chi_\lambda(U^\dagger \widehat{U}) \right|^2, \]
where $\chi_\lambda$ is the character of the irreducible unitary representation corresponding to $\lambda$. (Recall that the character of a unitary representation is the trace of the unitary's action; in this case, $\chi_\lambda(U^\dagger \widehat{U}) = \mathrm{Tr}(U^\dagger_\lambda\hat{U}_\lambda)$.)
\end{definition}

\begin{fact}[Covariance of $p$ \cite{yrc}]\label{fact:p_covariant} With this choice of probe state and measurement, $p$ is unitarily covariant, meaning $p(U_1 \widehat{U} U_2^\dagger | U_1 U U_2^\dagger) = p(\widehat{U} | U )$ for all unitaries $U_1,U_2$.
\end{fact}

\section{Results}

\subsection{The near-optimal protocol}\label{sec:near-optimal_protocol_defn}

Our protocol is based on the protocol in \cite{yrc}. The main new idea is to view the construction in \cite{yrc} as a member of a family of protocols defined by a parameter $r \in \mathbb{N}_+$. Different choices of $r$ lead to distinct supports $\mathsf{Y}_\text{probe}$ for the distribution $q$.

\begin{definition}[$\mathcal{A}(r,d,n,U)$] Our algorithm $\mathcal{A}$ for near-optimal unitary tomography takes as inputs known parameters $r,d$, and $n$ as well as an unknown $U \in \mathrm{U}(d)$. We then conduct the parallel unitary tomography protocol in \Cref{def:mo_protocol} with the probe state in \Cref{def:probe_state} and measurement in \Cref{def:protocol_measurement}. All that is left is to specify the distribution $q_\lambda$ which defines the probe state. 

We first define its support $\mathsf{Y}_\text{probe}$. $\mathsf{Y}_\text{probe}$ will consist of Young diagrams shaped like upside-down ``staircases'' constructed according to this rough recipe. First, use $d(d-1)(r+1)/2$ of the $n$ boxes to form a Young diagram whose consecutive row lengths each decrease by $r+1$ boxes. Then, using the remaining boxes, form a similarly shaped Young diagram whose first $d-1$ row lengths decrease by some constant determined by $n,d,$ and $r$. This constant is called $N$. The last row gap is not $N$ but instead $1$. Because boxes come in integer amounts, it may not be possible to form this second diagram exactly, but we take a close approximation. We may be left with a small number $n_0$ of leftover boxes after the approximation, which get distributed as evenly as possible across all rows. Adding these two diagrams together, we form a ``base'' diagram in $\mathsf{Y}_\text{probe}$ with consecutive row gaps of at least $r+N+1$ (except for the last row, which has a gap of $r+2$). Finally, we remove a limited number of boxes from the last row of the base diagram and redistribute them across the first $d-1$ rows. While doing so, we make sure that consecutive row gaps are still more than $r$ boxes after redistributing. $\mathsf{Y}_\text{probe}$ consists of these redistributed diagrams. For each $\lambda \in \mathsf{Y}_\text{probe}$, we define a vector $\tilde{\lambda}$ to encode the way its boxes were redistributed. The weight $q_\lambda$ associated to $\lambda \in \mathsf{Y}_\text{probe}$ is determined by $\tilde{\lambda}$.

Now we make this recipe precise. To that end, define the following constants:
\begin{align*}
    N &= \left\lfloor \frac{1}{3d-2}\left(\frac{2n}{d-1} + d - 2 - dr \right) \right\rfloor \\
    n_0 &= n - \frac{((3d-2)N-d+2+dr)(d-1)}{2}.
\end{align*}

Let $\mu_0$ be the flattest Young diagram in $\mathsf{Y}_{n_0}^d$; more precisely, it is the unique diagram satisfying that each difference between its consecutive rows is 0, except for possibly one difference of 1. Next, define $\lambda_0 \in \mathsf{Y}_{n-n_0}^d$ to have rows of length
\begin{equation} {\lambda_0}_j = N(2d-3) + (d-1)r + 1 - (N+r+1)(j-1), \quad 1 \leq j \leq d. \end{equation}
We introduce the notation $\lambda|_{d-1}$ to refer to the first $d-1$ rows of a Young diagram $\lambda$. Define
\begin{equation} \mathsf{Y}_\text{probe} = \{ \lambda \in \mathsf{Y}_n^d \text{ such that }  \lambda|_{d-1} = \mu_0|_{d-1} + \lambda_0|_{d-1} + \tilde\lambda \text{ for some } \tilde\lambda \in \{0,\ldots,N-1\}^{d-1} \}. \end{equation}
The last rows are uniquely determined by the fact that the total size is $n$; for a given $\tilde{\lambda}$, $\lambda_d = {\mu_0}_d + {\lambda_0}_d - \sum_{j<d}\tilde{\lambda}_j$. Note that $\lambda \in \mathsf{Y}_\text{probe}$ are in bijection with $\tilde\lambda \in \{0,\ldots,N-1\}^{d-1}$.

Compared to the choice of $\mathsf{Y}_\text{probe}$ in \cite{yrc}, our definition is formulated to widen the gap between each row of each Young diagram by an additional $r$ boxes. We define $r$ to be strictly positive, but if one were to set $r=0$ in the formulas above, one would recover the choice of $\mathsf{Y}_\text{probe}$ in \cite{yrc}. We omit $r=0$ from our definitions since we will later require $1/r$ to be defined.

The distribution $q$ over $\mathsf{Y}_\text{probe}$ is a member of the sine state family from \cite{HLSZY}. It is defined in terms of the $\tilde\lambda$ corresponding to $\lambda$,
\begin{equation} q_\lambda = \prod_{j=1}^{d-1}\frac{2^{2r}}{N\binom{2r}{r}}\sin^{2r}\left(\frac{\pi(2\tilde\lambda_j+1)}{2N}\right). \end{equation}
\end{definition}

\subsection{Proof of \texorpdfstring{\Cref{thm:near-optimal_tomography}}{Theorem~1}}\label{sec:pf_of_main_thm}

While we use a unitary tomography protocol similar to the one in~\cite{yrc}, we analyze our protocol's performance using a novel proof strategy. The authors of~\cite{yrc} use orthogonality of characters to bound the error of their protocol when measured with entanglement infidelity, a distance measure which they show admits a simple expression in terms of characters. Applying this approach with diamond distance does not yield a similarly tractable expression. Instead, our proof hinges on other representation-theoretic tools, namely the Weyl integral formula and the Fourier transform, to bound diamond-distance error.

The main technical step of \Cref{thm:near-optimal_tomography} is handled in the following proposition.

\begin{proposition}[Error variance bound]\label{prop:error_variance} Let $\bm\e = d_\diamond(\U(U),\U(\bm{\widehat{U}}))$ where $\bm{\widehat{U}}$ is the random variable output of $\mathcal{A}(r,d,n,U)$. Assume that $d \geq 2$ and that $r,n,$ and $d$ are valued such that $N > r$. Then $\mathbb{E}(\bm\e^2) \leq 64 \pi^2/3 \cdot (d-1)^{1/r} r^2/N^2$.
\end{proposition}

We defer the full proof of \Cref{prop:error_variance} to section \Cref{sec:variance_bound_proof}. In brief, the proof proceeds in four main steps:
\begin{enumerate}
    \item Using the Weyl integral formula, we shift from taking an expectation over the unitary group $\mathrm{U}(d)$ to taking an expectation over the $d$-dimensional torus $\mathbb{T}^d$. Let $\bm\alpha = (\bm\alpha_1,\ldots,\bm\alpha_d)$ be a random element of $\mathbb{T}^d$ sampled from the resulting probability distribution on $\mathbb{T}^d$. There is a function $\Delta:\mathbb{T}^d \rightarrow \mathbb{R}$ that satisfies
    \begin{equation}\bbE_{\mathrm{U}(d)}(\bm{\e}^2) \leq \frac{4}{3} \bbE_{\mathbb{T}^d}(\Delta(\bm\alpha)^2).\end{equation}
    \item We show that bounding $\bbE_{\mathbb{T}^d}(\Delta(\bm\alpha)^2)$ reduces to bounding $\bbE_{\mathbb{T}^d}(|e^{i\bm{\alpha}_{j_1}} - e^{i\bm{\alpha}_{j_2}}|^{2r})$, where $1 \leq j_1 < j_2 \leq d$ are arbitrary indices.
    \item Next, we prove that $\bbE_{\mathbb{T}^d}(|e^{i\bm{\alpha}_{j_1}} - e^{i\bm{\alpha}_{j_2}}|^{2r})$ is controlled by finite differences of a function $u:\mathbb{Z} \rightarrow \mathbb{R}$ which is determined by our choice of amplitudes in our probe state. The argument proceeds by rewriting the expectation over the torus as the $L^2$ norm of a function on the torus. This equals the $\ell^2$ norm of the function's Fourier transform, which we can relate to $u$.
    \item By analyzing the derivatives of a continuization of $u$, we give an upper bound on $u$'s finite differences. This is sufficient to prove our desired result about $\mathbb{E}(\bm\e^2)$ when combined with the previous steps.
\end{enumerate}

\begin{repeatedstatement}[Near-optimal parallel unitary process tomography]{thm:near-optimal_tomography} Let $r,d \in \mathbb{N}_+$ with $d \geq 2$ and let $\e \in (0,1)$. Let $U \in \mathrm{U}(d)$ be unknown. The protocol $\mathcal{A}(r,d,n,U)$ with query complexity $n$ set to
\begin{align*}
    n &= \frac{(d-1)((3d-2)N_* - d+2+dr)}{2} \\
    &\in O\left(\frac{rd^{2+1/(2r)}}{\e}\right)
\end{align*}
where $N_* = \left\lceil 8\pi r(d-1)^{1/(2r)}\e^{-1} \right\rceil$ solves the unitary process tomography problem \Cref{def:unitary_tomography_problem} with error parameters $(\e,1/3)$, and all queries of $\mathcal{U}(U)$ occur in parallel.
\end{repeatedstatement}
\begin{proof} Given the definition of $N$ and our choice of $n$, one can show that $N = N_*$. This means $N>r$, so \Cref{prop:error_variance} applies. From \Cref{prop:error_variance}, we know that $\mathbb{E}(\bm{\e}^2) \leq 64 \pi^2/3 \cdot (d-1)^{1/r} r^2/N^2$, where $\bm{\e} = d_\diamond(\U(\bm{\widehat{U}}),\U(U))$ and $\bm{\widehat{U}}$ is the output of $\mathcal{A}(r,d,n,U)$. Using $N = N_*$, we obtain $\mathbb{E}(\bm{\e}^2) \leq \e^2/3$. By Markov's inequality applied to $\e^2$,
\begin{equation}\mathbf{Pr}(\bm\e \leq \e) \geq 1-\frac{\bbE(\bm{\e}^2)}{\e^2} \geq 1 - \frac{1}{3}. \end{equation}
$\mathcal{A}(r,d,n,U)$ therefore solves the unitary process tomography problem with error parameters $(\e,1/3)$, and by definition of $\mathcal{A}(r,d,n,U)$, all queries of $\U(U)$ occur in parallel.
\end{proof}

\section{Discussion}

In this work, we reduce the gap between optimal sequential unitary process tomography and parallel unitary process tomography to a log factor. We do so using a variant of the protocols developed in~\cite{Bisio_2010} and~\cite{yrc}, employing a probe state tailored to the dimension of the unitary to be estimated.

Furthermore, we contribute new ideas in our proof method. The method of proof used in~\cite{yrc} to compute unitary tomography's query complexity with respect to entanglement infidelity does not readily generalize to diamond distance. Our proof strategy uses a different suite of tools, including the Weyl integral formula, Schur polynomials, and the Fourier transform, to analyze diamond distance error. These techniques have not previously been applied in conjunction with the tomography protocol framework in \Cref{def:mo_protocol}. Our novel proof strategy may be adaptable to a broader class of problems, such as allowing for analysis of other distance measures.

Finally, we note some futher questions this result raises. Given that the optimal query complexity of unitary tomography is (almost) the same for two extreme circuit regimes---parallel and sequential---it is worth investigating the complexity of this task in intermediate circuit shapes that are neither fully parallel nor fully sequential. More generally, we can ask which other types of learning tasks may be made parallel at no asymptotic cost to query complexity. Our result also confirms a difference between the tomography framework and the metrology framework for unitary estimation: While metrology exhibits a polynomial gap in query complexity between parallel and sequential circuits~\cite{IF,Yuan}, our result shows that the tomography framework we consider does not have such a gap~\cite{HKOT}. Seeking to understand the root cause of that difference could shed light on learning tasks beyond unitary process tomography.

\vspace{0.2in}
\noindent\textbf{Note added.---}An independent work~\cite{HLSZY2} will be concurrently posted on arXiv, presenting a parallel protocol for unitary tomography that differs from the one developed here and achieving a query complexity of $O(d^2/\e)$.

\section{Acknowledgments}

We thank Marco Fanizza, Dmitry Grinko, Entong He, Zihao Li, Antonio Anna Mele, Francesco Anna Mele and Yuxiang Yang for discussion. We acknowledge the support of the Natural Sciences and Engineering Research Council of Canada (NSERC), [DGECR-2025-00505, RGPIN-2025-04054], National Research Council of Canada (NRC), [AQC-217-1], and
Perimeter Institute for Theoretical Physics, a research institute supported in part by the Government of Canada through the Department of Innovation, Science and Economic Development Canada and by the Province of Ontario through the Ministry of Colleges and Universities. 

ChatGPT 5.6 was used in the development of part of the proofs presented in this paper, including the Fourier-transform-based proof technique and the suggestion of the probe state family. All model-generated arguments were independently verified and further developed by the authors, and the final manuscript was written entirely by the authors. The authors take full responsibility for the correctness of all mathematical statements and arguments.

\bibliographystyle{quantum}
\bibliography{references}

\onecolumn
\appendix

\section{Proof of \texorpdfstring{\Cref{prop:error_variance}}{Proposition~2}}\label{sec:variance_bound_proof}

\subsection{Preliminaries for variance bound}

Within this section, $d$ and $m$ represent positive integers. $\mathbb{T} = \mathbb{R}/2\pi\mathbb{Z}$ denotes the torus.

\begin{definition}[$C_c(\mathbb{Z}^m)$] The set of functions $\mathbb{Z}^m \rightarrow \mathbb{C}$ with finite support (or equivalently, compact support) is $C_c(\mathbb{Z}^m)$. 
\end{definition}

\begin{notation} We write elements of $C_c(\mathbb{Z}^m)$ with a small hat, as in $\hat{f}$. This is distinct from the wide hat over $\widehat{U}$ denoting an estimate. The symbols $\widehat{U},\widehat{W},$ and $\bm{\widehat{U}}$ with wide hats are not elements of $C_c(\mathbb{Z}^m)$, but all other symbols in this document with hats are.
\end{notation}

\begin{definition}[Trigonometric polynomial, $\mathcal{T}$] $f : \mathbb{T}^m \rightarrow \mathbb{C}$ is a \textit{trigonometric polynomial} if there exists $\hat{f} \in C_c(\mathbb{Z}^m)$ such that
\begin{equation} f(\alpha) = \sum_{k \in \mathbb{Z}^m} \hat{f}(k) e^{ik \cdot \alpha}. \end{equation}
($k \cdot \alpha$ is a dot product.) Note that $\hat{f}$ is unique when it exists. The set of trigonometric polynomials on $\mathbb{T}^m$ is $\mathcal{T}(\mathbb{T}^m)$.
\end{definition}

\begin{definition}[Fourier transform on $\mathbb{T}^m$] The \textit{Fourier transform} of a trigonometric polynomial $f \in \mathcal{T}(\mathbb{T}^m)$ is the unique function $\hat{f} \in C_c(\mathbb{Z}^m)$ such that $f(\alpha) = \sum_{k \in \mathbb{Z}} \hat{f}(k) e^{ik \cdot \alpha}$.
\end{definition}

\begin{definition}[$\ell^2(\mathbb{Z}^m)$ norm] Functions $\hat{f} \in C_c(\mathbb{Z}^m)$ have finite $\ell^2$ norms $\|\hat{f}\|_{\ell^2(\mathbb{Z}^m)}$ defined by $\|\hat f\|_{\ell^2(\mathbb{Z}^m)}^2 = \sum_{k \in \mathbb{Z}^m} |\hat{f}(k)|^2$.
\end{definition}

\begin{definition}[$L^2$ norm] Trigonometric polynomials $f \in \mathcal{T}(\mathbb{T}^m)$ have finite $L^2$ norms $\|f\|_{L^2(\mathbb{T}^m)}$ defined by $\|f\|_{L^2(\mathbb{T}^m)}^2 = 1/(2\pi)^d\int_{\mathbb{T}^m} \mathrm{d}\alpha\ |f(\alpha)|^2$. For $D \subseteq \mathbb{R}^m$ and $f : D \rightarrow \mathbb{C}$ for which the following integral is defined, the $L^2$ norm is given by $\|f\|_{L^2(D)}^2 = \int_D \mathrm{d}x\ |f(x)|^2$.
\end{definition}

\begin{fact}[Plancherel's theorem]\label{fact:plancherel} Let $f\in \mathcal{T}(\mathbb{T}^m)$ have Fourier transform $\hat{f} \in C_c(\mathbb{Z}^m)$. Then $\| f \|_{L^2(\mathbb{T}^m)} = \| \hat{f} \|_{\ell^2(\mathbb{Z}^m)}$.
\end{fact}

\begin{lemma}[Sums and products with $\ell^2$ norm]\label{lem:sum_and_product_norms} We have the following:
\begin{enumerate}
    \item Given functions $\hat{f}_1,\ldots,\hat{f}_l \in C_c(\mathbb{Z}^m)$ with pairwise disjoint supports, define $\hat{f} \in C_c(\mathbb{Z}^m)$ by $\hat{f}(k) = \hat{f}_1(k) + \cdots + \hat{f}_l(k)$. Then 
    \begin{equation} \| \hat{f} \|_{\ell^2(\mathbb{Z}^m)}^2 = \sum_{j=1}^l \| \hat{f}_j \|_{\ell^2(\mathbb{Z}^m)}^2. \end{equation}
    \item Given functions $\hat{g}_1,\ldots,\hat{g}_m \in C_c(\mathbb{Z})$, define $\hat{g} \in C_c(\mathbb{Z}^m)$ by $\hat{g}(k) = \hat{g}_1(k_1) \cdot \ldots \cdot \hat{g}_m(k_m)$. Then 
    \begin{equation} \| \hat{g} \|_{\ell^2(\mathbb{Z}^m)} = \prod_{j=1}^m \| \hat{g}_j \|_{\ell^2(\mathbb{Z})}. \end{equation}
\end{enumerate}
\end{lemma}
\begin{proof}
1. Since the supports of $\hat{f}_j$ are all pairwise disjoint, summing a quantity over all $k \in \bigcup_j\text{supp}(\hat{f}_j)$ is the same as summing that quantity over $k \in \text{supp}(\hat{f}_j)$ for each $j$, then adding the results. Therefore
\begin{align}
    \| \hat{f} \|_{\ell^2(\mathbb{Z}^m)}^2 &= \sum_{k \in \mathbb{Z}^m} | \hat{f}_1(k) + \cdots + \hat{f}_l(k) |^2 \\
    &= \sum_{k \in \bigcup_j\text{supp}(\hat{f}_j)} | \hat{f}_1(k) + \cdots + \hat{f}_l(k) |^2 \\
    &= \sum_{j=1}^l \sum_{k \in \text{supp}(\hat{f}_j)} | \hat{f}_1(k) + \cdots + \hat{f}_l(k) |^2 \\
    &= \sum_{j=1}^l \sum_{k \in \text{supp}(\hat{f}_j)} | 0 + \cdots + 0 + \hat{f}_j(k) + 0 + \cdots + 0 |^2 \\
    &= \sum_{j=1}^l \sum_{k \in \mathbb{Z}^m} | \hat{f}_j(k) |^2 \\
    &= \sum_{j=1}^l \| \hat{f}_j \|_{\ell^2(\mathbb{Z})}^2.
\end{align}

2. We have
\begin{align}
    \| \hat{g} \|_{\ell^2(\mathbb{Z}^m)}^2 &= \sum_{k \in \mathbb{Z}^m} | \hat{g}_1(k_1) \cdot \ldots \cdot \hat{g}_m(k_m) |^2 \\
    &= \sum_{k \in \mathbb{Z}^m} | \hat{g}_1(k_1) |^2 \cdot \ldots \cdot | \hat{g}_m(k_m) |^2 \\
    &= \sum_{k_1 \in \mathbb{Z}} \cdots \sum_{k_m \in \mathbb{Z}} | \hat{g}_1(k_1) |^2 \cdot \ldots \cdot | \hat{g}_m(k_m) |^2 \\
    &= \sum_{k_1 \in \mathbb{Z}} | \hat{g}_1(k_1) |^2 \cdot \ldots \cdot \sum_{k_m \in \mathbb{Z}} | \hat{g}_m(k_m) |^2 \\
    &= \prod_{j=1}^m \| \hat{g}_j \|_{\ell^2(\mathbb{Z})}^2.
\end{align}
\end{proof}

\begin{definition}[$e_j,S_j,D_j$] We let $e_j$ be the $j$th standard basis vector. Define the invertible linear \textit{shift operator} $S_j$ on $C_c(\mathbb{Z}^m)$ by $[S_j(\hat{f})](k) = \hat{f}(k-e_j)$. Define the linear \textit{forward and backward finite difference operators} $D_{+,j}$ and $D_{-,j}$ on $C_c(\mathbb{Z}^m)$ by $D_{\pm,j} = I - S_j^{\mp1}$.
\end{definition}

\begin{lemma}[Shift operators and finite difference operators]\label{lem:shift_and_fin_dif_ops} The following hold for all $j$:
\begin{enumerate}
    \item $S_j$ is unitary with respect to the $\ell^2(\mathbb{Z}^m)$ norm.
    \item $e^{i\alpha_j} \cdot f \in \mathcal{T}(\mathbb{T}^m)$ has Fourier transform $S_j(\hat{f}).$
    \item There is an operator $S'$ on $C_c(\mathbb{Z}^m)$ that is unitary with respect to the $\ell^2(\mathbb{Z}^m)$ norm satisfying $D_{+,j} = S' D_{-,j}$.
\end{enumerate}
\end{lemma}
\begin{proof}
1. This follows because summing over all lattice points $k \in \mathbb{Z}^m$ is the same as summing over all $k + e_j \in \mathbb{Z}^m$,
\begin{equation} \| S_j(\hat{f}) \|_{\ell^2(\mathbb{Z}^m)}^2 = \sum_{k \in \mathbb{Z}^m} |\hat{f}(k-e_j)|^2 = \sum_{k+e_j \in \mathbb{Z}^m} |\hat{f}(k)|^2 = \sum_{k \in \mathbb{Z}^m} |\hat{f}(k)|^2 = \| \hat{f} \|_{\ell^2(\mathbb{Z}^m)}^2. \end{equation}

2. We know $f(\alpha) = \sum_{k \in \mathbb{Z}^m} \hat{f}(k) e^{ik \cdot \alpha}$, so 
\begin{equation} e^{i\alpha_j}\cdot f(\alpha) = \sum_{k \in \mathbb{Z}^m} \hat{f}(k) e^{i(k+e_j) \cdot \alpha} = \sum_{k-e_j \in \mathbb{Z}^m} \hat{f}(k-e_j) e^{ik \cdot \alpha} = \sum_{k \in \mathbb{Z}^m} \hat{f}(k-e_j) e^{ik \cdot \alpha}. \end{equation}
The Fourier transform of $e^{i\alpha_j}\cdot f(\alpha)$ is therefore the function mapping $k$ to $\hat{f}(k-e_j)$, which is $S_j(\hat{f})$.

3. Since $S_j$ is unitary, so is $-S_j^{-1}$, and
\begin{equation} D_{+,j} = I - S_j^{-1} = -S_j^{-1}(I - S_j) = -S_j^{-1}D_{-,j}. \end{equation}
\end{proof}

\begin{definition}[$u$]\label{def:u} Let $n,r,d \in \mathbb{N}_+$ be fixed so that the constant $N$ is specified. Define $u \in C_c(\mathbb{Z})$ by 
\begin{equation} u(k) = \begin{cases}
    \frac{2^{r}}{\sqrt{N\binom{2r}{r}}}\sin^{r}\left(\frac{\pi(2k+1)}{2N}\right), &k \in \{0,\ldots,N-1\} \\
    0, &k \notin \{0,\ldots,N-1\}
\end{cases} \end{equation}
One can check that for $\tilde{\lambda} \in \{0,\ldots,N-1\}^{d-1}$ we have $\sqrt{q_\lambda} = u(\tilde{\lambda}_1) \cdot \ldots \cdot u(\tilde{\lambda}_{d-1})$. Also, as long as $N > r$, $u$ is normalized to $\|u\|_{\ell^2(\mathbb{Z})} = 1$.
\end{definition}

\begin{fact}[Schur polynomials and eigenangle expression for characters, chapter 24 of \cite{fulton2013representation}]\label{fact:Schur_polys} Let $\lambda$ be a Young diagram with at most $d$ rows and let $\chi_{\lambda}$ be the character of the irreducible representation of $\mathrm{U}(d)$ corresponding to $\lambda$. Let $\widehat{U}\in\mathrm{U}(d)$ have a tuple of eigenangles $\alpha = (\alpha_1,\ldots,\alpha_d)$. Then we can write $\chi_{\lambda}(\widehat{U})$ as a Schur polynomial of the eigenvalues $e^{i \alpha_1}, \ldots, e^{i \alpha_d}$. Namely, let $V(\alpha)$ be the Vandermonde matrix $V(\alpha)_{jk} = (e^{i\alpha_k})^{d-j}$ and let $M(\alpha, \lambda)$ be the matrix with entries $M(\alpha,\lambda)_{jk} = (e^{i\alpha_k})^{d-j+\lambda_j}$. We have
\begin{equation} \chi_{\lambda}(\widehat{U}) = \frac{\det(M(\alpha, \lambda))}{\det(V(\alpha))}. \end{equation}
This is indeed a polynomial in $e^{i \alpha_1}, \ldots, e^{i \alpha_d}$ because the denominator divides the numerator.
\end{fact}

\begin{lemma}[$\omega$ and eigenangle expression for $p(\widehat{U}|I)$]\label{lem:omega_facts} Let $\widehat{U} \in \mathrm{U}(d)$ and let $\alpha = (\alpha_1,\ldots,\alpha_d)$ be a tuple of $\widehat{U}$'s eigenangles. Let $\mathrm{S}_d$ denote the symmetric group of degree $d$. Fix $n,r \in \mathbb{N}_+$ so that the set $\mathsf{Y}_\textnormal{probe}$ is specified. Then the function 
\begin{align}
    \omega : \mathsf{Y}_\textnormal{probe} \times \mathrm{S}_d &\rightarrow \mathbb{Z}^d \\
    (\lambda, \tau) &\mapsto (d-\tau(j)+\lambda_{\tau(j)})_{j=1}^d
\end{align}
has the following properties:
\begin{enumerate}
    \item $\displaystyle p(\widehat{U}|I) = \left| \frac{1}{\det V(\alpha)}\sum_{\lambda\in \mathsf{Y}_{\text{probe}}(r)} \sum_{\tau \in \mathrm{S}_d} \textnormal{sgn}(\tau) \sqrt{q_\lambda} e^{i\alpha \cdot \omega(\lambda,\tau)} \right|^2$ (here $\alpha \cdot \omega(\lambda,\tau)$ denotes the dot product).
    \item $\omega(\lambda,\tau) = \tau^{-1}(\omega(\lambda,\textnormal{id}))$.
    \item The components of $\omega(\lambda,\textnormal{id})$ strictly decrease with gaps of more than $r$, $\omega(\lambda,\textnormal{id})_j - \omega(\lambda,\textnormal{id})_{j+1} > r$ for all $j = 1, \ldots, d-1$.
    \item $\omega$ is injective.
\end{enumerate}
\end{lemma}
\begin{proof}
1. \Cref{def:protocol_measurement} expresses $p(\widehat{U}|I)$ in terms of characters, and \Cref{fact:Schur_polys} expresses characters in terms of eigenangles, resulting in
\begin{align}
    p(\widehat{U}|I) &= \left| \sum_{\lambda \in \mathsf{Y}_\text{probe}} \sqrt{q_\lambda} \chi_\lambda(I^\dagger \widehat{U}) \right|^2 \\
    &= \left| \sum_{\lambda \in \mathsf{Y}_\text{probe}} \sqrt{q_\lambda} \frac{\det(M(\alpha, \lambda))}{\det(V(\alpha))} \right|^2 \\
    &= \left| \frac{1}{\det(V(\alpha))} \sum_{\lambda \in \mathsf{Y}_\text{probe}} \sqrt{q_\lambda} \sum_{\tau \in \mathrm{S}_d} \text{sgn}(\tau) \prod_{j=1}^d M(\alpha,\lambda)_{\tau(j)j} \right|^2 \\
    &= \left| \frac{1}{\det(V(\alpha))} \sum_{\lambda \in \mathsf{Y}_\text{probe}} \sum_{\tau \in \mathrm{S}_d} \sqrt{q_\lambda} \text{sgn}(\tau) e^{i \alpha \cdot \omega(\lambda,\tau)} \right|^2.
\end{align}

2. The $j$th component of $\tau^{-1}(\omega(\lambda,\text{id}))$ equals the $j$th component of $\omega(\lambda,\tau)$,
\begin{equation} \left(\tau^{-1}(\omega(\lambda,\text{id}))\right)_j = \omega(\lambda,\text{id})_{\tau(j)} = d-\tau(j) + \lambda_{\tau(j)} = \omega(\lambda,\tau)_{j}. \end{equation}

3. We have 
\begin{equation} \omega(\lambda,\text{id})_j - \omega(\lambda,\text{id})_{j+1} = d - j + \lambda_j - d + (j+1) - \lambda_{j+1} = 1 + \lambda_j - \lambda_{j+1}.\end{equation}
We know $\lambda \in \mathsf{Y}_\text{probe}$, so $\lambda$ corresponds to some $\tilde{\lambda} \in \{0,\ldots,N-1\}^{d-1}$. We proceed by cases on $j$. If $j < d-1$, then the formula $\lambda_j = (\lambda_0)_j + (\mu_0)_j + \tilde{\lambda}_j$ holds for both $j$ and $j+1$. By definition, $(\mu_0)_j - (\mu_0)_{j+1} \geq 0$, $\tilde{\lambda}_j - \tilde{\lambda}_{j+1} \geq -(N-1)$, and $(\lambda_0)_j - (\lambda_0)_{j+1} = N+1+r$, so $1 + \lambda_j - \lambda_{j+1} \geq 2 + r > r$, as desired. For the other case, when $j=d-1$, we find $\lambda_d = (\mu_0)_d + (d-1)(N-1) -\sum_{j=0}^{d-1} \tilde{\lambda}_j$ while $\lambda_{d-1} = (\mu_0)_{d-1} + (d-1)(N-1) + r + 2 + \tilde{\lambda}_{d-1}$,  implying $1 + \lambda_{d-1}-\lambda_d > r$.

4. Components of $\omega(\lambda,\text{id})$ strictly decrease for all $\lambda$. This means the ordering of components from least to greatest of $\omega(\lambda,\tau) = \tau^{-1}(\omega(\lambda,\text{id}))$ is determined only by $\tau$, not by $\lambda$. For $\omega(\lambda,\tau) = \omega(\lambda',\tau')$ to hold, both vectors must in particular share the same ordering of components from least to greatest, or in other words, we must have $\tau = \tau'$. Hence $\omega$ is injective if $\omega(-,\text{id})$ is injective. If $\omega(\lambda,\text{id}) = \omega(\lambda',\text{id})$, then $d - j + \lambda_j = d - j + \lambda'_{j}$ for all $j$, implying $\lambda = \lambda'$. Therefore $\omega(-,\text{id})$ is injective and $\omega$ is injective.
\end{proof}

\begin{fact}[Weyl integral formula, chapter 26 of \cite{fulton2013representation}]\label{fact:Weyl_integral_formula} Let $f:\mathrm{U}(d) \rightarrow \mathbb{C}$ be integrable and be a class function, meaning $f$ depends only on the eigenvalues of its input. Let $V(\alpha)$ be the Vandermonde matrix with entries $V_{jk}=(e^{i\alpha_k})^{d-j}$. The Weyl integral formula states
\begin{equation} \int_{\mathrm{U}(d)} \mathrm{d} U f(U) = \frac{1}{d!(2\pi)^d}\int_{\mathbb{T}^d} \mathrm{d}\alpha\ f\left(\text{diag}\left({e^{i\alpha}}\right)\right) \cdot |\det(V(\alpha))|^2. \end{equation}
\end{fact}

\begin{definition}[$\sigma(\alpha)$] Let $U\in \mathrm{U}(d)$ have spectrum $\{e^{i\alpha_1}, \ldots, e^{i\alpha_d}\}$ and let $\alpha = (\alpha_1, \ldots, \alpha_d)$ be a tuple of $U$'s eigenangles. We define $\sigma(\alpha) = 2\pi - \max_{j=1}^d(\alpha_j - \alpha_{j+1} \pmod{2\pi})$ (where $\alpha_{d+1}$ denotes $\alpha_1$) to be the angle measure of the minimal arc on the unit circle containing $U$'s spectrum.
\end{definition}

\begin{fact}[Diamond distance and unitary channels]\label{fact:unitary_diamond_dist} For unitaries $U_1,U_2 \in \mathrm{U}(d)$, we can express $d_\diamond(\U(U_1),\U(U_2))$ in terms of the spectrum $\{e^{i \alpha_1}, \ldots, e^{i \alpha_d}\}$ of $U_1^\dagger U_2$ and its spread $\sigma(\alpha)$ (appendix A of \cite{HKOT}):
\begin{equation} d_\diamond(\U(U_1),\U(U_2)) = 2\sin\frac{\min\{\sigma(\alpha), \pi\}}{2}. \end{equation}
Furthermore, $d_\diamond$ is \textit{unitarily covariant}~\cite{watrous2018theory}, meaning for all $W_1,W_2 \in \mathrm{U}(d)$,
\begin{equation} d_\diamond(\U(U_1),\U(U_2)) = d_\diamond(\U(W_1U_1W_2),\U(W_1U_2W_2)). \end{equation}
\end{fact}

\begin{definition}[$\Delta$] $\Delta$ takes in the eigenangles of a unitary and outputs the maximum distance between pairs of its eigenvalues,
\begin{align}
    \Delta : \mathbb{T}^d &\rightarrow \mathbb{R} \\
    (\alpha_1,\ldots,\alpha_d) &\mapsto \max_{j_1,j_2 \in \{1,\ldots,d\}} |e^{i\alpha_{j_1}}-e^{i\alpha_{j_2}}|.
\end{align}
\end{definition}

\begin{lemma}[$\Delta$ and diamond distance error]\label{lem:delta_bound}
    Let $U \in \mathrm{U}(d)$ have eigenangles $\alpha = (\alpha_1, \ldots, \alpha_d)$. Then $d_\diamond(\U(U), \U(I)) \leq 2/\sqrt{3}\cdot \Delta(\alpha)$.
\end{lemma}
\begin{proof}
    Since $U$ and $I$ are unitary, $d_\diamond(\U(U),\U(I)) = 2\sin(\min\{ \sigma(\alpha),\pi \}/2)$ where $\sigma(\alpha)$ is the arc length of the shortest arc containing $U$'s spectrum $\{ e^{i\alpha_1}, \ldots, e^{i\alpha_d} \}$ (\Cref{fact:unitary_diamond_dist}). 
    
    When $\sigma(\alpha) \leq \pi$, the maximum distance between eigenvalues is the distance between the eigenvalues at each end of the shortest arc containing $U$'s spectrum. Hence $\Delta(\alpha) = 2\sin(\sigma(\alpha)/2) = d_\diamond(\U(U),\U(I)).$ 
    
    When $\sigma(\alpha) \geq \pi$, consider the eigenvalues $z_1$ and $z_2$ at each end of the shortest arc containing $U$'s spectrum. There is a semicircle between $z_1$ and $-z_1$ contained in the arc, and a semicircle between $z_2$ and $-z_2$ contained in the arc. If there were no eigenvalues in the intersection of these semicircles, then the arc would not be the shortest one containing $U$'s spectrum, a contradiction. Therefore there exists an eigenvalue $z_3$ in the semicircles' intersection. The convex hull of $z_1, z_2,$ and $z_3$ contains the origin. This implies some pair of $z_1,z_2,$ and $z_3$ is separated by an angle of at least $2\pi/3$ radians, giving $\Delta(\alpha) \geq 2\sin(\pi/3) = \sqrt{3}/2 \cdot 2\sin(\pi/2)$. Hence $2/\sqrt{3} \cdot \Delta(\alpha) \geq d_\diamond(\U(U), \U(I))$.
\end{proof}

\subsection{Proof of variance bound}

Our proof of \Cref{prop:error_variance}, which states that our protocol satisfies $\mathbb{E}(\bm\e^2) \in O(d^{1/r} r^2 N^{-2})$, follows the steps outlined in \ref{sec:pf_of_main_thm}. We repeat those steps here:
\begin{enumerate}
    \item We shift from taking an expectation over the unitary group $\mathrm{U}(d)$ to expectation over the $d$-dimensional torus $\mathbb{T}^d$. Let $\bm\alpha = (\bm\alpha_1,\ldots,\bm\alpha_d)$ be a random element of $\mathbb{T}^d$ sampled from this probability distribution. Then the function $\Delta:\mathbb{T}^d \rightarrow \mathbb{R}$ satisfies
    \begin{equation}\bbE_{\mathrm{U}(d)}(\bm{\e}^2) \leq \frac{4}{3} \bbE_{\mathbb{T}^d}(\Delta(\bm\alpha)^2).\end{equation}
    \item Let $1 \leq j_1 < j_2 \leq d$ be arbitrary indices. We show that bounding $\bbE_{\mathbb{T}^d}(\Delta(\bm\alpha)^2)$ reduces to bounding $\bbE_{\mathbb{T}^d}(|e^{i\bm{\alpha}_{j_1}} - e^{i\bm{\alpha}_{j_2}}|^{2r})$,
    \begin{equation} \bbE_{\mathbb{T}^d}(\Delta(\bm\alpha)^2) \leq 4 \left( \frac{1}{d} \sum_{1 \leq j_1 < j_2 \leq d} \bbE_{\mathbb{T}^d}\left(|e^{i\bm{\alpha}_{j_1}} - e^{i\bm{\alpha}_{j_2}}|^{2r}\right) \right)^{1/r}. \end{equation}
    \item $\bbE_{\mathbb{T}^d}(|e^{i\bm{\alpha}_{j_1}} - e^{i\bm{\alpha}_{j_2}}|^{2r})$ is controlled by finite differences of $u$,
    \begin{equation} \bbE_{\mathbb{T}^d}(|e^{i\bm\alpha_{j_1}} - e^{i\bm\alpha_{j_2}}|^{2r}) \leq \left( \sum_{s=0}^r \binom{r}{s} \| D_+^s u \|_{\ell^2(\mathbb{Z})} \| D_+^{r-s} u \|_{\ell^2(\mathbb{Z})} \right)^2. \end{equation} 
    We prove this by rewriting the expectation over the torus as the $L^2$ norm of a function on the torus. This equals the $\ell^2$ norm of the function's Fourier transform, which we can relate to $u$.
    \item Analyzing a continuized and rescaled version of $u$, we show
    \begin{equation} \| D_+^s u \|_{\ell^2(\mathbb{Z})} \leq \left( \frac{\pi r}{N}\right)^s. \end{equation}
    This is sufficient to conclude our desired bound of $\bbE(\bm{\e}^2) \in O(r^2 d^{1/r} N^{-2})$ when combined with the previous steps.
\end{enumerate}

\begin{repeatedstatement}[Error variance bound]{prop:error_variance} Let $\bm\e = d_\diamond(\U(U),\U(\bm{\widehat{U}}))$ where $\bm{\widehat{U}}$ is the random variable output of $\mathcal{A}(r,d,n,U)$. Assume that $d \geq 2$ and that $r,n,$ and $d$ are valued such that $N > r$. Then $\mathbb{E}(\bm\e^2) \leq 64 \pi^2/3 \cdot (d-1)^{1/r} r^2/N^2$.
\end{repeatedstatement}

\begin{proof}
We wish to bound
\begin{equation} \mathbb{E}(\bm\e^2) = \int_{\mathrm{U}(d)} \mathrm{d}\widehat{U}\ p(\widehat{U}|U) \cdot d_\diamond(\U(\widehat{U}),\U(U))^2 \end{equation}
where $\mathrm{d}\widehat{U}$ is the Haar measure. We structure the proof according to the four steps above.

\underline{Step 1:} The Haar measure is translation invariant and $p$ and $d_\diamond$ are unitarily covariant (\Cref{fact:p_covariant}, \Cref{fact:unitary_diamond_dist}). That lets us assume $U=I$ without loss of generality. More precisely, shifting the variable of integration to $\widehat{W} = \widehat{U}U^{-1}$ yields
\begin{equation} \mathbb{E}(\bm\e^2) = \int_{\mathrm{U}(d)} \mathrm{d}\widehat{W}\ p(\widehat{W}|I) \cdot d_\diamond(\U(\widehat{W}),\U(I))^2. \end{equation} 

The integrand is a class function depending only on the eigenangles of $\widehat{W}$. That means we can use the Weyl integral formula (\Cref{fact:Weyl_integral_formula}) to integrate over the torus of eigenangles $\mathbb{T}^d$ instead of integrating over $\mathrm{U}(d)$. This requires expressing our integrand in terms of the eigenangles $\alpha = (\alpha_1, \ldots, \alpha_d)$ of $\widehat{W}$. We have
\begin{align}
    p(\widehat{W}|I) &= \left| \frac{1}{\det V(\alpha)}\sum_{\lambda\in \mathsf{Y}_\text{probe}} \sum_{\tau \in \mathrm{S}_d} \text{sgn}(\tau) \sqrt{q_\lambda} e^{i\alpha \cdot \omega(\lambda,\tau)} \right|^2 & \text{by \Cref{lem:omega_facts}} \\
    d_\diamond(\U(\widehat{W}),\U(I)) &= 2\sin\frac{\min \{ \sigma(\alpha), \pi \}}{2} & \text{by \Cref{fact:unitary_diamond_dist}} \\
    &\leq \frac{2}{\sqrt{3}} \Delta(\alpha) & \text{by \Cref{lem:delta_bound}}.
\end{align}
Define $A \in \mathcal{T}(\mathbb{T}^d)$ by 
\begin{equation}
A(\alpha) = \sum_{\lambda\in \mathsf{Y}_\text{probe}} \sum_{\tau \in \mathrm{S}_d} \text{sgn}(\tau) \sqrt{q_\lambda} e^{i\alpha \cdot \omega(\lambda,\tau)}    
\end{equation} 
so that $p(\widehat{W}|I) = |A(\alpha)/V(\alpha)|^2$. By the Weyl integral formula,
\begin{align}
    \mathbb{E}(\bm\e^2) &= \int_{\mathbb{T}^d} \mathrm{d}\alpha\ \frac{1}{d!(2\pi)^d} |\det V(\alpha)|^2 \left| \frac{A(\alpha)}{\det V(\alpha)} \right|^2 \cdot \left(2\sin\frac{\min\{\sigma(\alpha),\pi\}}{2}\right)^2\\
    &\leq \frac{4}{3} \int_{\mathbb{T}^d} \mathrm{d}\alpha\ \frac{1}{d!(2\pi)^d} | A(\alpha) |^2 \cdot \Delta(\alpha)^2.
\end{align}
The quantity $1/(d!(2\pi)^d) \cdot |A(\alpha)|^2$ integrates to $1$ over the torus, so we may regard it as a probability distribution on $\mathbb{T}^d$. The right hand side then becomes $4/3 \cdot \mathbb{E}(\Delta(\bm\alpha)^2)$ where $\bm{\alpha}$ is drawn from that distribution. This completes step 1.

\underline{Step 2:} We begin by showing that for all $\alpha \in \mathbb{T}^d$,
\begin{equation}\label{eq:Delta_r_ineq}
    \Delta(\alpha)^2 \leq 4\left(\frac{1}{d} \sum_{1 \leq j_1 < j_2 \leq d} |e^{i\alpha_{j_1}} - e^{i\alpha_{j_2}}|^{2r} \right)^{1/r}.
\end{equation}
By definition of $\Delta$, there is some pair $j_a,j_b$ for which $|e^{i\alpha_{j_a}} - e^{i\alpha_{j_b}}| = \Delta(\alpha)$. For the $d-2$ other indices $j$, the triangle inequality implies one of $|e^{i\alpha_{j}} - e^{i\alpha_{j_a}}|$ or $|e^{i\alpha_{j}} - e^{i\alpha_{j_b}}|$ is at least $\Delta(\alpha)/2$. We find
\begin{equation} \sum_{1 \leq j_1 < j_2 \leq d} |e^{i\alpha_{j_1}} - e^{i\alpha_{j_2}}|^{2r} \geq (d-2) \cdot \left( \frac{\Delta(\alpha)}{2}\right)^{2r} + |e^{i\alpha_{j_a}} - e^{i\alpha_{j_b}}|^{2r} \geq d \left( \frac{\Delta(\alpha)}{2} \right)^{2r}. \end{equation}
Raising both sides to $1/r$ yields \Cref{eq:Delta_r_ineq}. 

Taking expectations of both sides of \Cref{eq:Delta_r_ineq} and applying Jensen's inequality finishes step 2.

\underline{Step 3:} We start by expressing $\mathbb{E}(|e^{i\alpha_{j_1}} - e^{i\alpha_{j_2}}|^{2r})$ as an $L^2(\mathbb{T}^d)$ norm,
\begin{equation} \mathbb{E}(|e^{i\alpha_{j_1}} - e^{i\alpha_{j_2}}|^{2r}) = \frac{1}{d!}\cdot \frac{1}{(2\pi)^d} \int_{\mathbb{T}^d} \mathrm{d}\alpha\  |A(\alpha)|^2 \cdot |e^{i\alpha_{j_1}} - e^{i\alpha_{j_2}}|^{2r} = \frac{1}{d!} \left\| (e^{i\alpha_{j_1}} - e^{i\alpha_{j_2}})^r \cdot A \right\|_{L^2(\mathbb{T}^d)}^2. \end{equation}
Next, we rewrite this as an $\ell^2$ norm in the Fourier domain. Let $\hat{A} \in C_c(\mathbb{Z}^d)$ be the Fourier transform of $A$. Under the Fourier transform, multiplicative factors of $e^{i\alpha_j}$ map to shift operators $S_j$ (\Cref{lem:shift_and_fin_dif_ops}). Then by Plancherel's theorem (\Cref{fact:plancherel}),
\begin{equation}\label{eq:applying_Plancherel} \frac{1}{d!} \left\| (e^{i\alpha_{j_1}} - e^{i\alpha_{j_2}})^r \cdot A \right\|_{L^2(\mathbb{T}^d)}^2 = \frac{1}{d!} \left\| [(S_{j_1} - S_{j_2})^r] (\hat{A}) \right\|_{\ell^2(\mathbb{Z}^d)}^2.\end{equation}

To simplify this quantity, we will express $[(S_{j_1} - S_{j_2})^r] (\hat{A})$ as a sum of $d!$ functions with pairwise disjoint supports. By definition, $A(\alpha) = \sum_{\lambda\in \mathsf{Y}_\text{probe}} \sum_{\tau \in \mathrm{S}_d} \text{sgn}(\tau) \sqrt{q_\lambda} e^{i\alpha \cdot \omega(\lambda,\tau)}$. Since $\omega$ is injective (\Cref{lem:omega_facts}), no two summands share the same frequency vector $\omega(\lambda,\tau)$. Therefore the coefficient on $e^{i\alpha \cdot k}$ in $A(\alpha)$ is
\begin{equation} \hat{A}(k) = \begin{cases}
    \text{sgn}(\tau) \sqrt{q_\lambda}, &\text{if } \exists \lambda\in \mathsf{Y}_\text{probe} \text{ and } \tau\in \mathrm{S}_d \text{ such that } k = \omega(\lambda,\tau) \\
    0, & \text{otherwise}.
\end{cases}\end{equation}

Now we express $\hat{A}$ as a sum of $d!$ functions. Let $K_\text{id}$ be the set of frequencies corresponding to the identity permutation, $K_\text{id} = \{ k \in \mathbb{Z}^d\ |\ \exists \lambda \in \mathsf{Y}_\text{probe} \text{ s.t. } k = \omega(\lambda,\text{id}) \}$. Define $\hat{B}$ to be $\hat{A}$ with support truncated to $K_\text{id}$, $\hat{B} = \hat{A} \cdot \mathds{1}_{K_\text{id}}$. Since $\omega(\lambda, \tau) = \tau^{-1}(\omega(\lambda,\text{id}))$ (\Cref{lem:omega_facts}), we have 
\begin{equation}
\hat{A} = \sum_{\tau \in \mathrm{S}_d} \text{sgn}(\tau)\hat{B}\circ \tau. 
\end{equation}

Next, we show these $d!$ summands have disjoint supports after being acted on by $(S_{j_1} - S_{j_2})^r$. The components of $\omega(\lambda,\text{id})$ strictly decrease with gaps of more than $r$, $\omega(\lambda,\text{id})_j - \omega(\lambda,\text{id})_{j+1} > r$ for all $j = 1, \ldots, d-1$ (\Cref{lem:omega_facts}). This means vectors $k \in \text{supp}(\hat{B})$ satisfy the same property, $k_j - k_{j+1} > r$ for all $j = 1, \ldots, d-1$. Acting on $\hat{B}$ with $(S_{j_1} - S_{j_2})^r$ yields a sum of scalar multiples of terms of the form $\pm S_{j_1}^{r_1}S_{j_2}^{r_2}\hat{B}$ where $0 \leq r_1,r_2 \leq r$. Vectors $k'$ in the supports of these terms are of the form $k' = k + r_1e_{j_1} + r_2e_{j_2}$ for some $k \in \text{supp}(\hat{B})$. Since gaps between consecutive components of $k$ are all at least $r$, the gaps between consecutive components of $k'$ are all at least $(r+1) - \max\{r_1,r_2\} > 0$ (since $j_1 \neq j_2$). This means for all $j_1,j_2,r_1,$ and $r_2$, $\text{supp}(\pm S_{j_1}^{r_1}S_{j_2}^{r_2}\hat{B})$ lies in the sector of $\mathbb{Z}^d$ consisting of vectors whose components strictly decrease. Calling that sector $Z = \{k \in \mathbb{Z}^d\ |\ k_1 > \cdots > k_d \}$, we conclude $\text{supp}([(S_{j_1} - S_{j_2})^r](\hat{B})) \subset Z$ for all $j_1,j_2$. Then for any $\tau \in \mathrm{S}_d$, 
\begin{align}
    \text{supp}\left([(S_{j_1} - S_{j_2})^r](\text{sgn}(\tau)\hat{B} \circ \tau)\right) &= \text{supp}\left([(S_{\tau(j_1)} - S_{\tau(j_2)})^r](\text{sgn}(\tau)\hat{B}) \circ \tau\right) \\
    &= \tau^{-1}\left(\text{supp}\left([(S_{\tau(j_1)} - S_{\tau(j_2)})^r](\text{sgn}(\tau)\hat{B})\right)\right) \\
    &= \tau^{-1}\left(\text{supp}\left([(S_{\tau(j_1)} - S_{\tau(j_2)})^r](\hat{B})\right)\right) \\
    &\subset \tau^{-1}(Z).
\end{align}
Since $\{\tau^{-1}(Z)\}_{\tau \in \mathrm{S}_d}$ are pairwise disjoint, the supports of $[(S_{j_1} - S_{j_2})^r](\text{sgn}(\tau)\hat{B} \circ \tau)$ are pairwise disjoint, as claimed.

Returning to the norm in \Cref{eq:applying_Plancherel},
\begin{align}
    \frac{1}{d!} \left\| [(S_{j_1} - S_{j_2})^r] (\hat{A}) \right\|_{\ell^2(\mathbb{Z}^d)}^2 &= \frac{1}{d!} \left\| \sum_{\tau \in \mathrm{S}_d} [(S_{j_1} - S_{j_2})^r](\text{sgn}(\tau)\hat{B}\circ \tau) \right\|_{\ell^2(\mathbb{Z}^d)}^2 \\
    &= \frac{1}{d!} \sum_{\tau \in \mathrm{S}_d} \left\| [(S_{j_1} - S_{j_2})^r](\text{sgn}(\tau)\hat{B}\circ \tau) \right\|_{\ell^2(\mathbb{Z}^d)}^2 \\
    &= \frac{1}{d!} \sum_{\tau \in \mathrm{S}_d} \left\| [(S_{\tau(j_1)} - S_{\tau(j_2)})^r](\text{sgn}(\tau)\hat{B})\circ \tau\right\|_{\ell^2(\mathbb{Z}^d)}^2 \\
    &= \frac{1}{d!} \sum_{\tau \in \mathrm{S}_d} \left\| [(I - S_{\tau(j_1)}^{-1}S_{\tau(j_2)})^r](\hat{B}) \right\|_{\ell^2(\mathbb{Z}^d)}^2 \\
    &= \frac{1}{d(d-1)} \sum_{\substack{1 \leq i_1 \leq d \\ 1 \leq i_2 \leq d \\ i_1 \neq i_2}} \left\| [(I - S_{i_1}^{-1}S_{i_2})^r](\hat{B}) \right\|_{\ell^2(\mathbb{Z}^d)}^2. 
\end{align}
The second line follows from summands having disjoint support (\Cref{lem:sum_and_product_norms}). The penultimate line holds because shift operators, multiplication by $\text{sgn}(\tau)$, and precomposition with $\tau$ are all unitary operations on $\ell^2(\mathbb{Z}^d)$ (\Cref{lem:shift_and_fin_dif_ops}).

$[(I - S_{i_1}^{-1}S_{i_2})^r](\hat{B})$ is a function from $\mathbb{Z}^d \rightarrow \mathbb{R}$. By considering a reparametrized version with domain $\mathbb{Z}^{d-1}$, we can relate this function to $u$. Consider the parametrization map $\tilde{\lambda} \mapsto \lambda$ extended to all of $\mathbb{Z}^{d-1}$. This is given by the injective affine map $\iota: \mathbb{Z}^{d-1} \rightarrow \mathbb{Z}^{d}, l \mapsto \mu_0 + \lambda_0 + l \cdot (e_1 - e_d, \ldots, e_{d-1} - e_d)$. Next, we name the map that brings a diagram $\lambda$ to a frequency vector in $\text{supp}(\omega(\lambda,\text{id}))$. We call it $\kappa: \mathbb{Z}^d \rightarrow \mathbb{Z}^d$ and define it as $\kappa(k) = k + (d-1, d-2, \ldots, 0)$. We now reparametrize $\hat{B}$ by defining $\hat{b} : \mathbb{Z}^{d-1} \rightarrow \mathbb{R}$ to be $\hat{b} = \hat{B} \circ \kappa \circ \iota$. We have $\hat{b}(l) = \sqrt{q_{\iota(l)}}$ for $l \in \{0,\ldots,N-1\}^{d-1}$ and $\hat{b}(l) = 0$ otherwise; in other words, $\hat{b}(l) = u(l_1) \cdot \ldots \cdot u(l_{d-1})$ (by definition of $u$ and $q_\lambda$). 

Using the reparametrization, we can express the $\ell(\mathbb{Z}^d)$ norm of $[(I - S_{i_1}^{-1}S_{i_2})^r](\hat{B})$ as the $\ell(\mathbb{Z}^{d-1})$ norm of a different function. The operator $S_{i_1}^{-1}S_{i_2}$ acting on $\hat{B}$ corresponds to the following operator $T_{i_1i_2}$ acting on $\hat{b}$,
\begin{align}
    S_{i_1}^{-1}S_{i_2}(\hat{B}) \circ \kappa \circ \iota = T_{i_1i_2}(\hat{b}) = \begin{cases}
        S_{i_1}^{-1}(\hat{b}), &i_1 < i_2 = d \\
        S_{i_2}(\hat{b}), &i_2 < i_1 = d \\
        S_{i_1}^{-1}S_{i_2}(\hat{b}), &\text{otherwise.}
    \end{cases}
\end{align}
Also, for each $0 \leq s \leq r$, we have that $\text{supp}([S_{i_1}^{-1}S_{i_2}]^{s}(\hat{B})) = \text{supp}(\hat{B}) - s(e_{i_1} - e_{i_2})$ is contained in $\text{im}(\kappa \circ \iota)$. This holds because $\text{supp}(\hat{B}) \subset \text{im}(\kappa \circ \iota)$ and because the shift $- s(e_{i_1} - e_{i_2})$ can be produced by altering the input to $\iota$ as follows:
\begin{align}
    \iota(l) - s(e_{i_1} - e_d) &= \iota(l - se_{i_1}), &\text{ for } i_1 < d\\
    \iota(l) - s(e_d - e_{i_2}) &= \iota(l + se_{i_2}), &\text{ for } i_2 < d \\
    \iota(l) - s(e_{i_1} - e_{i_2}) &= \iota(l - se_{i_1} + se_{i_2}), &\text{ for } i_1,i_2<d.
\end{align}
Hence $\text{supp}\left([(I - S_{i_1}^{-1}S_{i_2})^r](\hat{B})\right) \subset \text{im}(\kappa \circ \iota)$ for all $i_1\neq i_2$. That allows us to convert from the $\ell^2$ norm on $\mathbb{Z}^d$ to the $\ell^2$ norm on $\mathbb{Z}^{d-1}$,
\begin{align}
    \left\| [(I - S_{i_1}^{-1}S_{i_2})^r](\hat{B}) \right\|_{\ell^2(\mathbb{Z}^d)}^2 &= \sum_{k \in \mathbb{Z}^d} \left| \left([(I - S_{i_1}^{-1}S_{i_2})^r](\hat{B})\right)(k) \right|^2 \\
    &= \sum_{k \in \text{im}(\kappa \circ \iota)} \left| \left([(I - S_{i_1}^{-1}S_{i_2})^r](\hat{B})\right)(k) \right|^2 \\
    &= \sum_{l \in \mathbb{Z}^{d-1}} \left| \left([(I - S_{i_1}^{-1}S_{i_2})^r](\hat{B})\right)(\kappa \circ \iota(l)) \right|^2 \\
    &= \sum_{l \in \mathbb{Z}^{d-1}} \left| \left([(I - T_{i_1i_2})^r](\hat{b})\right)(l) \right|^2 \\
    &= \left\| [(I - T_{i_1i_2})^r](\hat{b}) \right\|_{\ell^2(\mathbb{Z}^{d-1})}^2
\end{align}
where the third line holds by injectivity of $\kappa \circ \iota$. 

Finally, we use the fact that $\hat{b}(l) = u(l_1) \cdot \ldots \cdot u(l_{d-1})$ to bound this norm in terms of finite differences of $u$. Functions $\mathbb{Z}^{d-1} \rightarrow \mathbb{C}$ expressible as a product of functions $\hat{f_i}: \mathbb{Z} \rightarrow \mathbb{C}$ on each input coordinate satisfy $\| \hat{f} \|_{\ell^2(\mathbb{Z}^{d-1})} = \| \hat{f}_1 \|_{\ell^2(\mathbb{Z})} \cdot \ldots \cdot \| \hat{f}_{d-1} \|_{\ell^2(\mathbb{Z})}$ (\Cref{lem:sum_and_product_norms}). $\hat{b}$ is of this form with function $u$ on each coordinate. $[(I - T_{i_1i_2})^r](\hat{b})$ is a sum of functions of this form with $u$ on most coordinates, but with different functions on coordinates $i_1$ and $i_2$. We now bound $\| [(I - T_{i_1i_2})^r](\hat{b}) \|_{\ell^2(\mathbb{Z}^{d-1})}$ by proceeding by cases on $i_1$ and $i_2$. When $i_1 < i_2 = d$, we know $I - T_{i_1i_2} = D_{+,i_1}$. Recalling that $N>r$, which ensures $\| u \|_{\ell^2(\mathbb{Z})}=1$ (\Cref{def:u}), we have
\begin{equation} \left\| [(I - T_{i_1i_2})^r](\hat{b}) \right\|_{\ell^2(\mathbb{Z}^{d-1})} = \| D_{+}^r u \|_{\ell^2(\mathbb{Z})} \prod_{i \in \{1,\ldots,d\} \setminus i_1} \| u \|_{\ell^2(\mathbb{Z})} = \| D_+^r u \|_{\ell^2(\mathbb{Z})}. \end{equation}
When $i_2 < i_1 = d$,
\begin{equation} \left\| [(I - T_{i_1i_2})^r](\hat{b}) \right\|_{\ell^2(\mathbb{Z}^{d-1})} = \| D_-^r u \|_{\ell^2(\mathbb{Z})} \prod_{i \in \{1,\ldots,d\} \setminus i_2} \| u \|_{\ell^2(\mathbb{Z})} = \| D_-^r u \|_{\ell^2(\mathbb{Z})}. \end{equation}
Finally, when neither $i_1$ nor $i_2$ equals $d$, we can write $I - T_{i_1i_2} = I - S_{i_1}^{-1}S_{i_2} = (I - S_{i_1}^{-1})S_{i_2} + (I - S_{i_2}) = S_{i_2}D_{+,i_1} + D_{-,i_2}$. Since $S_{i_2}D_{+,i_1}$ and $D_{-,i_2}$ commute, $(I - T_{i_1i_2})^r = \sum_{s=0}^r \binom{r}{s} (S_{i_2}D_{+,i_1})^{s} D_{-,i_2}^{r-s}$. Applying the triangle inequality,
\begin{align}
    \left\| [(I - T_{i_1i_2})^r](\hat{b}) \right\|_{\ell^2(\mathbb{Z}^{d-1})} &= \left\| \sum_{s=0}^r \binom{r}{s} \left[(S_{i_2}D_{+,i_1})^{s} D_{-,i_2}^{r-s}\right] (\hat{b})\right\|_{\ell^2(\mathbb{Z}^{d-1})} \\
    &\leq \sum_{s=0}^r \binom{r}{s} \left( \| (SD_+)^s u \|_{\ell^2(\mathbb{Z})} \| D_-^{r-s} u \|_{\ell^2(\mathbb{Z})} \prod_{i \in \{1,\ldots,d\} \setminus \{i_1,i_2\}} \| u \|_{\ell^2(\mathbb{Z})} \right).
\end{align}
$S$ is unitary, and forward and backward difference operators are unitarily related, $D_+ = (-S^{-1})D_-$ (\Cref{lem:shift_and_fin_dif_ops}), so
\begin{equation} \left\| [(I - T_{i_1i_2})^r](\hat{b}) \right\|_{\ell^2(\mathbb{Z}^{d-1})} \leq \sum_{s=0}^r \binom{r}{s} \| D_+^{s} u \|_{\ell^2(\mathbb{Z})} \| D_+^{r-s} u \|_{\ell^2(\mathbb{Z})}. \end{equation}
We conclude that in all cases, $\left\| [(I - T_{i_1i_2})^r](\hat{b}) \right\|_{\ell^2(\mathbb{Z}^{d-1})} \leq \sum_{s=0}^r \binom{r}{s} \| D_+^{s} u \|_{\ell^2(\mathbb{Z})} \| D_+^{r-s} u \|_{\ell^2(\mathbb{Z})}$. Hence
\begin{align}
    \mathbb{E}(|e^{i\alpha_{j_1}} - e^{i\alpha_{j_2}}|^{2r}) &= \frac{1}{d(d-1)} \sum_{\substack{1 \leq i_1 \leq d \\ 1 \leq i_2 \leq d \\ i_1 \neq i_2}} \left\| [(I - T_{i_1i_2})^r](\hat{b}) \right\|_{\ell^2(\mathbb{Z}^{d-1})}^2 \\
    &\leq \left(\sum_{s=0}^r \binom{r}{s} \| D_+^{s} u \|_{\ell^2(\mathbb{Z})} \| D_+^{r-s} u \|_{\ell^2(\mathbb{Z})}\right)^2,
\end{align}
completing step 3.

\underline{Step 4:} Let $x_m = (m + 1/2)/N$. Define $f \in \mathcal{T}(0,1)$ by $f(x) = \sin^r(\pi x)$. Let $F : \mathbb{R} \rightarrow \mathbb{R}$ be the extension of $f$ by 0 to all of $\mathbb{R}$. One can check that $u$ is proportional to $F$ evaluated at $x_m$, $u(m) = \sqrt{1/N}/\| f \|_{L^2(0,1)} \cdot F(x_m)$. Our goal for this step will be to bound finite differences of $u$ in terms of derivatives of $f$.

Define the operator $D_{1/N}$ on functions $\mathbb{R} \rightarrow \mathbb{R}$ by $[D_{1/N}(g)](x) = g(x+1/N) - g(x)$. Up to a sign, its action on $F$ corresponds to $D_+$ acting on $u$, 
\begin{equation}
    \frac{\sqrt{1/N}}{\| f \|_{L^2(0,1)}} \cdot [D_{1/N}(F)](x_m) = [-D_+(u)](m).
\end{equation} 

Next, we examine derivatives of $F$. For each $s \in \{ 0, \ldots, r-1\}$, the $s$th one-sided derivatives of $f$ at $0^+$ and $1^-$ are both 0. Hence $F^{(s)}$ is a continuous, piecewise smooth function when $0 \leq s \leq r-1$. If $s = r$, $F^{(r)}$ has discontinuities at $x=0$ and $x=1$; however, since $F^{(r-1)}$ is continuous, we may integrate $F^{(r)}$ as usual. In particular, when $0 \leq s \leq r$, we have $\| f^{(s)} \|_{L^2(0,1)}=\| F^{(s)} \|_{L^2(\mathbb{R})}$. 

When $s = 0$, we have $\| D_+^s(u) \|_{\ell^2(\mathbb{Z})} = \| u \|_{\ell^2(\mathbb{Z})} = 1 \leq (\pi r/N)^s$. For the remainder of this step, we focus on $1 \leq s \leq r$. By the fundamental theorem of calculus, $[D_{1/N}(F)](x) = F(x + 1/N) - F(x) = \int_0^{1/N} \mathrm{d}t\ F'(x+t)$. Applying this $s$ times,
\begin{equation} [D_{1/N}^s(F)](x_m) = \int_{[0,1/N]^s} \mathrm{d}t_1 \cdots \mathrm{d}t_s\ F^{(s)}(x_m + t_1 + \cdots + t_s). \end{equation}
Let $y = t_2 + \cdots + t_s$. Note that for a fixed $y$, the intervals $\{[x_m + y,x_m + y + 1/N]\}_{m \in \mathbb{Z}}$ tile $\mathbb{R}$. This means
\begin{equation} \sum_{m \in \mathbb{Z}} \int_0^{1/N} \mathrm{d}t_1\ | F^{(s)}(x_m + t_1 + y) |^2 = \int_\mathbb{R} \mathrm{d}t_1\ | F^{(s)}(t_1 + y) |^2 = \| F^{(s)} \|_{L^2(\mathbb{R})}^2. \end{equation}
Summing over $m$ and applying Cauchy-Schwarz to our expression for $[D_{1/N}^s(F)](x_m)$,
\allowdisplaybreaks
\begin{align}
    \sum_{m \in \mathbb{Z}} \left|[D_{1/N}^s(F)](x_m)\right|^2 &\leq \sum_{m \in \mathbb{Z}} \frac{1}{N^s} \int_{[0,1/N]^s} \mathrm{d}t_1 \cdots \mathrm{d}t_s\ | F^{(s)}(x_m + t_1 + \cdots + t_s) |^2 \\
    &= \frac{1}{N^s} \int_{[0,1/N]^{s-1}} \mathrm{d}t_2 \cdots \mathrm{d}t_s\ \sum_{m \in \mathbb{Z}} \int_0^{1/N} \mathrm{d}t_1\ | F^{(s)}(x_m + t_1 + y) |^2 \\
    &= \frac{1}{N^{2s-1}} \| F^{(s)} \|_{L^2(\mathbb{R})}^2 \\
    \nonumber \\
    \implies \| D_+^s u \|_{\ell^2(\mathbb{Z})}^2 &= \sum_{m \in \mathbb{Z}} | D_+^s u(m)|^2 \\
    &= \sum_{m \in \mathbb{Z}} \left| \frac{\sqrt{1/N}} {\| f \|_{L^2(0,1)}} \cdot [(-1)^sD_{1/N}^s(F)](x_m) \right|^2 \\
    &\leq \frac{1}{N^{2s}} \cdot \frac{\| F^{(s)} \|_{L^2(\mathbb{R})}^2}{\| f \|_{L^2(0,1)}^2} \\
    &= \left( \frac{1}{N^{s}} \cdot \frac{\| f^{(s)} \|_{L^2(0,1)}}{\| f \|_{L^2(0,1)}} \right)^2.
\end{align}
\allowdisplaybreaks

Direct computation shows that $f(x) = \sin^r(\pi x)$ satisfies $\|f^{(s)}\|_{L^2(0,1)} \leq (\pi r)^s \|f\|_{L^2(0,1)}$. Therefore $\| D_+^s(u) \|_{\ell^2(\mathbb{Z})} \leq (\pi r/N)^s$. This completes step 4.

We now put the steps together. Since $\| D_+^s(u) \|_{\ell^2(\mathbb{Z})} \leq (\pi r/N)^s$ for all $0 \leq s \leq r$ (step 4), we have $\sum_{s=0}^r \binom{r}{s} \| D_+^s(u) \|_{\ell^2(\mathbb{Z})} \| D_+^{r-s}(u) \|_{\ell^2(\mathbb{Z})} \leq (2\pi r/N)^r$. Then
\begin{align}
    &\bbE_{\mathrm{U}(d)}(\bm{\e}^2) \leq \frac{4}{3} \bbE_{\mathbb{T}^d}(\Delta(\bm\alpha)^2) &\text{(step 1)}\\
    &\leq \frac{16}{3} \left( \frac{1}{d} \sum_{1 \leq j_1 < j_2 \leq d} \bbE_{\mathbb{T}^d}\left(|e^{i\bm{\alpha}_{j_1}} - e^{i\bm{\alpha}_{j_2}}|^{2r}\right) \right)^{1/r} &\text{(step 2)}\\
    &\leq \frac{16}{3} \left( \frac{1}{d} \sum_{1 \leq j_1 < j_2 \leq d} \left(\sum_{s=0}^r \binom{r}{s} \| D_+^s(u) \|_{\ell^2(\mathbb{Z})} \| D_+^{r-s}(u) \|_{\ell^2(\mathbb{Z})}\right)^2 \right)^{1/r} &\text{(step 3)} \\
    &\leq \frac{16}{3} \left( (d-1) \left(\frac{2\pi r}{N}\right)^{2r} \right)^{1/r} &\text{(step 4)} \\
    &= \frac{64 \pi^2}{3} \cdot \frac{(d-1)^{1/r} r^2}{N^2}.
\end{align}
\end{proof}

\end{document}